\documentclass[reqno]{amsart}
\usepackage{amssymb}
\usepackage{hyperref}
\usepackage{booktabs}
\usepackage{multirow}
\usepackage{tikz}

\newtheorem{theorem}{Theorem}
\newtheorem{proposition}[theorem]{Proposition}
\newtheorem{lemma}[theorem]{Lemma}
\newtheorem{corollary}[theorem]{Corollary}

\theoremstyle{remark}

\newtheorem*{note}{Note}

\numberwithin{equation}{section}

\begin{document}

\title[Elliptic Ruijsenaars operators on bounded partitions]
{Elliptic Ruijsenaars difference operators on bounded partitions}

\author{Jan Felipe  van Diejen}

\address{
Instituto de Matem\'aticas, Universidad de Talca,
Casilla 747, Talca, Chile}

\email{diejen@inst-mat.utalca.cl}

\author{Tam\'as G\"orbe}

\address{School of Mathematics, University of Leeds, Leeds LS2 9JT,
UK}

\email{T.Gorbe@leeds.ac.uk}


\date{May 2021}

\begin{abstract} 
By means of a truncation condition on the parameters,
the elliptic Ruijsenaars difference operators are restricted onto a finite lattice of points encoded by bounded partitions.
A corresponding orthogonal basis of joint eigenfunctions is constructed in terms of polynomials on the joint spectrum. In the trigonometric limit, this recovers
the diagonalization of the truncated Macdonald difference operators by a finite-dimensional basis of Macdonald polynomials.
\end{abstract}

\maketitle

\section{Introduction}\label{sec1}
Ruijsenaars' commuting difference operators constitute the quantum integrals for a relativistic deformation of the elliptic Calogero-Moser-Sutherland system
\cite{rui:complete}. For integral values of the coupling parameter, the spectral problem for these difference operators has been fruitfully analyzed
within the framework of the algebraic Bethe Ansatz
 \cite{bil:algebraic,hou-sas-yan:eigenvalues,kom:essential,kom:ruijsenaars}; the underlying solutions of 
the Yang-Baxter equation incorporating Ruijsenaars' difference operators within the algebraic Bethe Ansatz formalism had been developed for this purpose in
 \cite{has:ruijsenaars,kom-hik:quantum,fel-var:elliptic}. For arbitrary positive values of the coupling parameter, an orthogonal basis of eigenfunctions for
the Ruijsenaars difference operators can be generated by means of a Cauchy-type Hilbert-Schmidt kernel 
 \cite{rui:hilbert}. Very recently, remarkably explicit formulas for these eigenfunctions were presented in \cite{lan-nou-shi:construction}
 (cf. also \cite{lan:explicit} and references therein for a direct construction of the eigenfunctions from the Hilbert-Schmidt kernel 
 in the nonrelativistic limit).   It is expected that a complete solution of the eigenvalue problem for the elliptic Ruijsenaars difference operators gives rise to an intricate elliptic counterpart 
  \cite{eti-kir:affine,lan-nou-shi:construction,mir-mor-zen:duality} of Macdonald's 
ubiquitous theory of symmetric orthogonal polynomials \cite[Chapter VI]{mac:symmetric}.

Motivated by the state of the art sketched above,  we will restrict the elliptic Ruijsenaars difference operators in the present note to functions
supported on a finite lattice in the center-of-mass configuration space (inside the Weyl alcove of the $\mathfrak{sl}(n+1;\mathbb{C})$ Lie algebra). The corresponding classical integrable particle system has a compact phase space given by the complex projective space and has been studied in \cite{feh-gor:trigonometric}.
In the  trigonometric  limit,  the lattice quantum particle model of interest was investigated in \cite{die-vin:quantum,gor-hal:quantization} as  the quantization  of a corresponding classical integrable system from \cite{rui:action,feh-klu:new}. 
The hyperbolic counterpart of our quantum model involves particles moving on an infinite lattice subject to a dynamics
 that is governed by a scattering matrix that factorizes into two-particle contributions \cite{rui:factorized}; when placed in an integrability-preserving external field,
 extra factors encoding one-particle contributions emerge in the scattering matrix \cite{die-ems:spectrum}. At the elliptic level of present concern, a detailed mathematical study of the quantum eigenvalue problem for the two-particle Ruijsenaars difference operator on the finite lattice was performed recently in \cite{die-gor:elliptic}.
 
Specifically, below we will adjust the real period of the elliptic functions in terms of the positive coupling parameter so as to
truncate the elliptic Ruijsenaars operators. This entails a corresponding commuting system of discrete difference operators, which are normal in a Hilbert space of functions supported in a finite lattice of shifted dominant weights associated with the $\mathfrak{sl}(n+1;\mathbb{C})$ Lie algebra.  With the aid of a standard labelling of the pertinent dominant weights
by means of (bounded) partitions, we build  an orthogonal basis
of joint eigenfunctions in terms of polynomials evaluated on the spectrum; the polynomials in question are uniquely determined by a recurrence stemming from the eigenvalue equations. In this approach, the spectral theorem for commuting normal operators in finite dimension entails the orthogonality and Pieri rules for these polynomials.
In the trigonometric limit, our construction recovers the unitary diagonalization of finitely truncated Macdonald operators by means of Macdonald polynomials from \cite{die-vin:quantum}.

The material is organized as follows. In Section \ref{sec2} the elliptic Ruijsenaars quantum lattice model restricted to bounded partitions is presented;  details on how to retrieve this lattice model from Ruijsenaars' commuting difference operators by discretization are supplied in Appendix \ref{app} (at the end of this note).
Section \ref{sec3}  promotes the corresponding discrete elliptic Ruijsenaars operators to commuting normal operators in an appropriate Hilbert space.
 An orthogonal basis of joint eigenfunctions given by polynomials evaluated on the spectrum is constructed in Section \ref{sec4}.
 In  Section \ref{sec5} we compare the trigonometric limit of the present construction with the corresponding diagonalization in terms of Macdonald polynomials stemming from \cite{die-vin:quantum}. This comparison reveals some salient features of the joint spectrum at the elliptic level, which we briefly highlight in the form of an epilogue in Section \ref{sec6}.

 \begin{note}
 Throughout  elliptic functions will be expressed in terms of the following rescaled Jacobi theta function
 \begin{subequations}
 \begin{equation}\label{scaled-theta}
[z]=[z;p]=\dfrac{\vartheta_1(\frac{\alpha}{2}z ;p)}{\frac{\alpha}{2}\vartheta'_1(0;p)}\qquad (z\in \mathbb{C},\, \alpha >0, \ 0<p<1),
\end{equation}
with
\begin{align}\label{theta}
 \vartheta_1(z;p)&=2\sum_{l\geq 0} (-1)^l p^{(l+\frac{1}{2})^2}\sin(2l+1)z,\\
&=2p^{1/4}\sin(z)\prod_{l\geq 1} (1-p^{2l})(1-2p^{2l}\cos(2z)+p^{4l} ). \nonumber
\end{align}
\end{subequations}
Hence $\vartheta^\prime_1(0;p)= 2p^{1/4} (p^2;p^2)_\infty^3$, where
\begin{equation}\label{q-Pochhammer}
(z;p)_\infty=\lim_{k\to \infty} (z;p)_k\quad\text{with}\quad  (z;p)_k=\prod_{0\leq l <k} (1-zp^l)
\end{equation}
(and the convention $(z;p)_0=1$).
The relation to the Weierstrass sigma function associated with the period  lattice $\Omega=2\omega_1 \mathbb{Z}+2\omega_2 \mathbb{Z}$ reads
(cf. e.g. \cite[\S 23.6(i)]{olv-loz-boi-cla:nist}):
\begin{equation}\label{jw-conversion}
[z]= \sigma (z) e^{-\tfrac{\zeta (\omega_1)}{2\omega_1}z^2}, 
\end{equation}
where $\alpha=\frac{\pi}{\omega_1}$, $p=e^{i\pi \tau}$ with $\tau= \frac{\omega_2}{\omega_1} $, and $\zeta(z)=\frac{\sigma^\prime(z)}{\sigma(z)}$.
 \end{note}

\section{Lattice Ruijsenaars models on partitions}\label{sec2}

\subsection{Discrete Ruijsenaars operators on partitions}
Let us recall that a \emph{partition} 
\begin{equation*}
\lambda=(\lambda_1,\lambda_2,\lambda_3,\ldots ,\lambda_\ell)=(\lambda_1,\lambda_2,\lambda_3,\ldots ,\lambda_\ell,0,\ldots )
\end{equation*}
 of \emph{weight} $|\lambda|=\sum_{j\in\mathbb{N}} \lambda_j$
consists of a weakly decreasing sequence of nonnegative integers of which only a finite number are allowed to be positive.  These positive terms are called parts and their number gives the \emph{length}
$\ell=\ell (\lambda)= | \{ j\in\mathbb{N} \mid \lambda_j >0\}| $  of the partition (where $| \cdot |$ refers to the cardinality of the set in question).
Following standard conventions
one writes  for two partitions  that $\lambda\subset\mu$ if $\lambda_j\leq\mu_j$ for all $j\in\mathbb{N}$; in this situation $\mu$ and $\lambda$ are said to differ by a \emph{vertical $r$-strip} $\theta$ if $|\mu |=|\lambda|+r$ and $\theta_j=\mu_j-\lambda_j\in \{ 0,1\}$ for all $j\in\mathbb{N}$. 
It is important to emphasize that a vertical $r$-strip $\theta$ is only a partition if $\theta=1^r$, where
\begin{equation*}
m^r=(\underbrace{m,\ldots,m}_{r})
\end{equation*}
refers to the rectangular partition with $r$ \emph{parts} of \emph{size} $m$ each.  
For  $n\in\mathbb{N}$, we will single out the set of partitions with at most $n$ parts:
\begin{equation}\label{partitions:n}
\Lambda^{(n)}=\{  \lambda\in\Lambda \mid \ell (\lambda) \leq n \} ,
\end{equation}
where $\Lambda$ stands for the set of all partitions. Moreover, occasionally we will indicate
 for a vertical strip $\theta$  (committing a slight abuse of notation) that $|\theta|=r$ and $\theta\subset 1^n$ (even if  $\theta$ is not a partition) so as to specify that one deals with an $r$-strip for which
$\theta_j=0$ when $j>n$.

After these preliminaries, we are now in the position to define the following discrete Ruijsenaars operators acting in the space of complex lattice functions
\begin{equation*}
\mathcal{C}(\Lambda^{(n)})=\{ \lambda\stackrel{f}{\to} f_\lambda \in\mathbb{C} \mid \lambda\in\Lambda^{(n)} \}
\end{equation*}
 for $1\leq r \leq n$:
\begin{subequations}
\begin{equation}\label{Dr:a}
(D_r  f)_\lambda =  \sum_{\substack{\lambda\subset\mu\subset\lambda+1^{n+1} \\ |\mu|=|\lambda |+r}}   B_{\mu/\lambda}   f_{\underline\mu} 
\qquad (f\in\mathcal{C}(\Lambda^{(n)}),\  \lambda\in\Lambda^{(n)} ),
\end{equation}
where
\begin{equation}
\underline{\mu}=(\mu_1-\mu_{n+1},\mu_2-\mu_{n+1},\ldots,\mu_n-\mu_{n+1} )
\end{equation}
($\in\Lambda^{(n)}$) and
\begin{equation}\label{Dr:c}
B_{\mu/\lambda} = \prod_{1\leq j<k\leq n+1}  \frac{[\lambda_j-\lambda_k+\mathrm{g}(k-j+\theta_j-\theta_k)]}{[\lambda_j-\lambda_k+\mathrm{g}(k-j)]} 
\quad\text{with}\ \theta=\mu-\lambda.
\end{equation}
\end{subequations}
Here $\mathrm{g}\in \mathbb{R}$ denotes a coupling parameter that is momentarily assumed to be generic such  that 
\begin{equation}\label{gc}
\forall\lambda\in\Lambda^{(n)}: \quad  \prod_{1\leq j<k\leq n+1}{\textstyle  \sin\frac{\alpha}{2}\big( \lambda_j-\lambda_k+\mathrm{g}(k-j)\bigr) } \neq 0
\end{equation}
 (which is the case e.g. if we pick $\mathrm{g}\in\mathbb{R}$ such that
$j\mathrm{g}\not\in \mathbb{Z}+\frac{2\pi}{\alpha}\mathbb{Z}$ for $j=1,\ldots ,n$). Notice that the genericity condition in \eqref{gc} ensures that the denominator
of $B_{\mu/\lambda} $ \eqref{Dr:c} does not vanish (cf. Eqs. \eqref{scaled-theta}, \eqref{theta}).

\begin{proposition}[Commutativity]\label{integrability:prp}
The discrete Ruijsenaars operators $D_1,\ldots ,D_n$ \eqref{Dr:a}--\eqref{Dr:c} commute in $\mathcal{C}(\Lambda^{(n)})$.
\end{proposition}

\begin{proof}
The operators $D_1, \ldots, D_n$ arise as discretizations of Ruijsenaars' commuting difference operators and therewith inherit their commutativity (cf. Appendix \ref{app} for further details).
\end{proof}

\subsection{Finite-dimensional truncation  on bounded partitions}
From now on we fix $n,m\in\mathbb{N}$ and pick
\begin{equation}\label{truncation}
\boxed{\textstyle \alpha=\frac{2\pi}{(n+1)\mathrm{g}+m}\quad\text{with}\ \mathrm{g}>0.} 
\end{equation}
Let
\begin{equation}
\Lambda^{(n,m)}=\{  \lambda\in\Lambda^{(n)}\mid \lambda\subset m^n \} .
\end{equation}
Notice that $| \Lambda^{(n,m)}|= \binom{n+m}{n}$.

\begin{lemma}[Truncation]\label{truncation:lem}
Let $\lambda\in\Lambda^{(n,m)}$  and  $\mu\in\Lambda^{(n+1)}$ such that $\theta=\mu-\lambda$ is a vertical $r$-strip.
Then---for  parameters in accordance with Eq. \eqref{truncation}---the coefficient $B_{\mu/\lambda}$ \eqref{Dr:c}  is \emph{positive} iff $\underline{\mu}\in\Lambda^{(n,m)}$
and \emph{vanishes} iff $\underline{\mu}\not\in\Lambda^{(n,m)}$.
\end{lemma}

\begin{proof}
For $\lambda\in\Lambda^{(n,m)}$ and  $1\leq j<k\leq n+1$, one has that
\begin{equation*}
0<\mathrm{g} \leq \lambda_j-\lambda_k+(k-j)\mathrm{g} \leq m +n\mathrm{g}  < {\textstyle \frac{2\pi}{\alpha}} .
\end{equation*}
Hence, it is clear from the product formula in Eqs. \eqref{scaled-theta}, \eqref{theta} that the corresponding factors in the denominator of $B_{\mu/\lambda}$ \eqref{Dr:c} 
are all positive.  
In the same way one deduces  that
\begin{equation*}
0\leq \lambda_j-\lambda_k +\mathrm{g} (k-j+\theta_j-\theta_k)\leq  {\textstyle \frac{2\pi}{\alpha}} 
\end{equation*}
in this situation.
The corresponding factors in the numerator of  $B_{\mu/\lambda}$ \eqref{Dr:c} 
are therefore all nonnegative.  Vanishing factors in the numerator occur iff the bounds of the interval are reached.
One reaches the lower bound zero iff $k=j+1$ with  $\theta_j=0$, $\theta_{j+1}=1$ and $\lambda_j=\lambda_{j+1}$ for some $j\in \{1,\ldots , n\}$; this can only happen  iff $\mu_j-\mu_{j+1}=\theta_j-\theta_{j+1}<0$, which contradicts  our assumption that $\mu\in\Lambda^{(n+1)}$.   
The upper bound $\frac{2\pi}{\alpha}$  is reached on the other hand iff  $j=1$,  $k=n+1$, with $\theta_1=1$, $\theta_{n+1}=0$ and $\lambda_1=m$ (since $\lambda_{n+1}=0$); this happens iff
$\mu_1-\mu_{n+1}=\lambda_1+\theta_1-\theta_{n+1}=m+1$, i.e. iff $\underline{\mu}\not\in\Lambda^{(n,m)}$.  
\end{proof}

It is immediate from Lemma \ref{truncation:lem} that---for  parameters subject to the truncation condition in \eqref{truncation}---the $\binom{n+m}{n}$-dimensional subspace
\begin{equation*}
\mathcal{C}(\Lambda^{(n,m)})=\{\lambda\stackrel{f}{\to} f_\lambda\in\mathbb{C} \mid \lambda\in\Lambda^{(n,m)}\} 
\end{equation*}
of functions in $\mathcal{C}(\Lambda^{(n)})$ supported on the finite lattice $\Lambda^{(n,m)}$ of bounded partitions is stable with respect to the action of $D_r$ \eqref{Dr:a}--\eqref{Dr:c}:
\begin{equation}\label{Drf}
(D_r  f)_\lambda =  \sum_{\substack{\lambda\subset\mu\subset\lambda+1^{n+1},\,  |\mu|=|\lambda|+r \\ \text{s.t.}\, \underline\mu\in\Lambda^{(n,m)}}}   B_{\mu/\lambda}   f_{\underline\mu} 
\qquad (f \in\mathcal{C}(\Lambda^{(n,m)}),\, \lambda\in \Lambda^{(n,m)}) .
\end{equation}
Notice in this connection that Lemma \ref{truncation:lem} ensures that all coefficients 
$B_{\mu/\lambda}$  \eqref{Dr:c}  in
$D_r$ \eqref{Drf} remain regular for parameter values in the domain \eqref{truncation}; as a consequence, we will
drop the transitory genericity condition in Eq. \eqref{gc} from now on (by analytic continuation) unless explicitly stated otherwise.

\begin{corollary}[Commutativity]\label{commutativity:cor}
For  parameters in accordance with Eq. \eqref{truncation}, the finite discrete Ruijsenaars operators $D_1,\ldots ,D_n$ \eqref{Drf} commute in $\mathcal{C}(\Lambda^{(n,m)})$.
\end{corollary}

\section{Hilbert space}\label{sec3}

\subsection{Inner product}
We consider \emph{elliptic weights} on $\Lambda^{(n,m)}$ of the form
\begin{equation}\label{Delta}
\Delta_\lambda =\prod_{1\leq j<k\leq n+1}
{\textstyle
\frac{[\lambda_j-\lambda_k+(k-j)\mathrm{g}]}{[(k-j)\mathrm{g} ]}\frac{[(k-j+1)\mathrm{g} ]_{\lambda_j-\lambda_k}}{[1+(k-j-1)\mathrm{g} ]_{\lambda_j-\lambda_k}} } ,
\end{equation}
where $[z]_k$, $k=0,1,2,\ldots$ denotes the \emph{elliptic factorial}
\begin{equation*}
[z]_k=\prod_{0\leq l<k} [z+l]\quad\text{with}\  [z]_0=1.
\end{equation*}

\begin{lemma}[Positivity]\label{positivity:lem}
For parameters in accordance with Eq.  \eqref{truncation} and any  $\lambda\in \Lambda^{(n,m)}$,
 the elliptic weight $\Delta_\lambda$ \eqref{Delta} is \emph{positive}.
\end{lemma}

\begin{proof}
In view of the product formula for $[z]$ \eqref{scaled-theta}, \eqref{theta} it is immediate from the assumptions that all factors in $\Delta_\lambda$ remain positive, since
\begin{equation*}
0<(k-j)\mathrm{g} \leq \lambda_j-\lambda_k+(k-j)\mathrm{g}  <{\textstyle \frac{2\pi}{\alpha} } 
\end{equation*}
if $1\leq j<k\leq n+1$, and
\begin{equation*}
0< l+ (k-j+1)\mathrm{g}  <{\textstyle \frac{2\pi}{\alpha}},\qquad  0< l+1+ (k-j-1)\mathrm{g}  <{\textstyle \frac{2\pi}{\alpha}} 
\end{equation*}
if  $0\leq l<\lambda_j-\lambda_k$.
\end{proof}
The positivity of the weights promotes $\mathcal{C} (\Lambda^{(n,m)})$ to a $\binom{n+m}{n}$-dimensional
Hilbert space $\ell^2(\Lambda^{(n,m)},\Delta )$ via the inner product
\begin{equation}\label{inner-product}
\langle f,g\rangle_\Delta=\sum_{\lambda\in\Lambda^{(n,m)}}  f_\lambda \overline{ g_\lambda }\Delta_\lambda
\qquad \bigl( f,g\in \ell^2(\Lambda^{(n,m)},\Delta ) \bigr).
\end{equation}

\subsection{Self-adjointness}
The elliptic weights $\Delta_\lambda$ \eqref{Delta}
obey a recurrence relation governed by the coefficients $B_{\mu/\lambda}$ \eqref{Dr:c}.

\begin{lemma}[Recurrence for Elliptic Weights]\label{Delta-rec:lem}
The elliptic weights $\Delta_\lambda$ \eqref{Delta} satisfy the following recurrence relation:
\begin{equation}\label{Delta-rec}
B_{\lambda+\theta/\lambda}\Delta_\lambda=B_{\underline{\mu}+\theta^c/\underline{\mu}}\Delta_{\underline{\mu}},
\end{equation}
 for any $\lambda\in\Lambda^{(n,m)}$  and $\mu\in\Lambda^{(n+1)}$ such that $\theta=\mu-\lambda$ is a vertical $r$-strip and $\underline{\mu}\in\Lambda^{(n,m)}$. Here
$\theta^c$ denotes the vertical $(n+1-r)$-strip such that $\theta+\theta^c=1^{n+1}$.
\end{lemma}

\begin{proof}
An elementary computation reveals that

\begin{equation*}
\begin{aligned}
\Delta_{\underline\mu}&=\prod_{1\leq j<k\leq n+1}
{\textstyle \frac{[(k-j+1)\mathrm{g} ]_{\lambda_j-\lambda_k+\theta_j-\theta_k}}{[(k-j)\mathrm{g} ]_{\lambda_j-\lambda_k+\theta_j-\theta_k}}\frac{[(k-j)\mathrm{g} +1]_{\lambda_j-\lambda_k+\theta_j-\theta_k}}{[(k-j-1)\mathrm{g} +1]_{\lambda_j-\lambda_k+\theta_j-\theta_k}} }\\
&= \Delta_\lambda
\prod_{\substack{1\leq j<k\leq n+1\\\theta_j-\theta_k=1}}
{\textstyle \frac{[\lambda_j-\lambda_k+(k-j+1)\mathrm{g} ]}{[\lambda_j-\lambda_k+(k-j-1)\mathrm{g} +1]}\frac{[\lambda_j-\lambda_k+(k-j)\mathrm{g} +1]}{[\lambda_j-\lambda_k+(k-j)\mathrm{g} ]} }\\
&\quad\times\prod_{\substack{1\leq j<k\leq  n+1\\\theta_j-\theta_k=-1}}
{\textstyle \frac{[\lambda_j-\lambda_k+(k-j-1)\mathrm{g} ]}{[\lambda_j-\lambda_k+(k-j+1)\mathrm{g} -1]}\frac{[\lambda_j-\lambda_k+(k-j)\mathrm{g} -1]}{[\lambda_j-\lambda_k+(k-j)\mathrm{g} ]} } .
\end{aligned}
\end{equation*}
Multiplication by
\begin{align*}
B_{\underline\mu+\theta^c/\underline\mu } &=
\prod_{1\leq j<k\leq n+1 }
{\textstyle \frac{[\mu_j-\mu_k+(k-j+\theta_k-\theta_j)\mathrm{g} ]}{[\mu_j-\mu_k+(k-j)\mathrm{g} ]} }
\\
&=
\prod_{\substack{1\leq j<k\leq n+1 \\ \theta_j-\theta_k=1}}
{\textstyle \frac{[\lambda_j-\lambda_k+(k-j-1)\mathrm{g} +1]}{[\lambda_j-\lambda_k+(k-j)\mathrm{g} +1]}} \prod_{\substack{1\leq j<k\leq n+1 \\ \theta_j-\theta_k=-1}}
{\textstyle \frac{[\lambda_j-\lambda_k+(k-j+1)\mathrm{g} -1]}{[\lambda_j-\lambda_k+(k-j)\mathrm{g} -1]} }
\end{align*}
yields
\begin{align*}
&\Delta_\lambda \prod_{\substack{1\leq j<k\leq n+1 \\\theta_j-\theta_k=1}}
{\textstyle \frac{[\lambda_j-\lambda_k+(k-j+1)\mathrm{g} ]}{[\lambda_j-\lambda_k+(k-j)\mathrm{g} ]} }
\prod_{\substack{1\leq j<k\leq n+1 \\\theta_j-\theta_k=-1}} 
{\textstyle  \frac{[\lambda_j-\lambda_k+(k-j-1)\mathrm{g} ]}{[\lambda_j-\lambda_k+(k-j)\mathrm{g} ]} } \\
&=\Delta_\lambda \prod_{1\leq j<k\leq n+1 }
{\textstyle \frac{[\lambda_j-\lambda_k+(k-j+\theta_j-\theta_k)\mathrm{g} ]}{[\lambda_j-\lambda_k+(k-j)\mathrm{g} ]} }
=\Delta_\lambda B_{\lambda+\theta /\lambda}.
\end{align*}
\end{proof}

With the aid of the recurrence in Lemma \ref{Delta-rec:lem}, one readily computes the adjoint of $D_r$ \eqref{Drf} in $\ell^2(\Lambda^{(n,m)},\Delta )$.
To this end it is convenient to rewrite the action of $D_r$ \eqref{Drf} in the form
\begin{equation}\label{Drf2}
(D_r f)_\lambda = \sum_{\theta\subset 1^{n+1},\,  |\theta |=r}  B_{\lambda+\theta/\lambda} f_{\underline{\lambda+\theta}}\,
\qquad (f \in \ell^2 (\Lambda^{(n,m)},\Delta),\, \lambda\in \Lambda^{(n,m)}) .
\end{equation}
In this formula the sum on the RHS is meant to be over all $r$-strips $\theta$ with $\theta_j=0$ for $j>n+1$,
where $B_{\lambda+\theta/\lambda}=0$ unless $\underline{\lambda+\theta}\in\Lambda^{(n,m)}$ by (the proof of)  Lemma \ref{truncation:lem}.

\begin{proposition}[Adjoint]\label{adjoint:prp}
For parameters in accordance with Eq. \eqref{truncation},
the operators $D_r$ and $D_{n+1-r}$ are each others adjoints in the Hilbert space $\ell^2(\Lambda^{(n,m)},\Delta )$, i.e.
\begin{equation}\label{adjoint}
\forall f,g\in   \ell^2(\Lambda^{(n,m)},\Delta ): \qquad      \langle D_rf,g\rangle_\Delta=\langle f,D_{n+1-r}g\rangle_\Delta .
\end{equation}
\end{proposition}

\begin{proof}
Successive manipulations hinging on Lemma \ref{truncation:lem}, Eq. \eqref{Drf2} and  Lemma \ref{Delta-rec:lem}  sustain that
\begin{align*}
&\langle D_rf,g\rangle_\Delta=\sum_{\lambda\in\Lambda^{(n,m)}}(D_rf)_\lambda \overline{g_\lambda}\Delta_\lambda = \\
&\sum_{\substack{\lambda\in\Lambda^{(n,m)}\\ \theta\subset 1^{n+1},\,  |\theta |=r}}     \left( B_{\lambda+\theta/\lambda}\, f_{\underline{\lambda+\theta}} \right)
\overline{g_\lambda }\Delta_\lambda         =\sum_{\substack{\mu\in\Lambda^{(n,m)}\\ \theta\subset 1^{n+1},\,  |\theta |=n+1-r}}    
 f_{\mu}  \left( B_{\mu+\theta/\mu}  \,  \overline{g_{\underline{\mu+\theta}} }  \right)  \Delta_\mu  \\
&=\sum_{\mu\in\Lambda^{(n,m)} }f_\mu \overline{(D_{n+1-r}g)_\mu }\Delta_\mu =\langle f,D_{n+1-r}g\rangle_\Delta .
\end{align*}
\end{proof}

\begin{corollary}[Self-adjointness]\label{self-adjoint:cor}
The discrete difference operators
\begin{subequations}
\begin{equation}
C_r={\textstyle \frac{1}{2}}(D_r+D_{n+1-r})\qquad (r=1,\ldots,\lfloor {\textstyle \frac{n+1}{2}} \rfloor ) 
\end{equation}
and
\begin{equation}
 S_r={\textstyle \frac{1}{2\mathrm{i}}}(D_r-D_{n+1-r}) \qquad (r=1,\ldots,\lfloor {\textstyle \frac{n}{2}} \rfloor ) 
\end{equation}
\end{subequations}
provide $n$ commuting Ruijsenaars operators that are self-adjoint in $\ell^2(\Lambda^{(n,m)},\Delta )$.
\end{corollary}

\begin{note}
If one replaces the truncation condition in Eq. \eqref{truncation} with
the genericity condition from Eq. \eqref{gc}, then the definition of the
elliptic weight $\Delta_\lambda$ \eqref{Delta} actually makes sense for any $\lambda\in\Lambda^{(n)}$ (as the zeros of the factors in the denominator are avoided).
The recurrence in Eq. \eqref{Delta-rec} holds in this situation for any $\lambda\in\Lambda^{(n)}$ and $\mu\in\Lambda^{(n+1)}$ such that $\theta=\mu-\lambda$ is a vertical $r$-strip, even though the positivity is now no longer guaranteed.
\end{note}

\section{Eigenfunctions}\label{sec4}

\subsection{Diagonalization}

The main result of this note consist of the following theorem, which describes the construction of an orthogonal basis of joint eigenfunctions for the Ruijsenaars model
on the finite lattice of bounded partitions; the (values of the) eigenfunctions  in question are expressed by means of polynomials in terms of the corresponding eigenvalues. 
To describe these polynomials we recur to following partial order on $\Lambda^{(n)}$ (which stems from the dominance ordering of the $\mathfrak{sl}(n+1;\mathbb{C})$ dominant weights via the bijection in. Eq. \eqref{bijection}):
\begin{equation}
\forall \lambda,\mu\in\Lambda^{(n)}:\quad \lambda\leq\mu \Leftrightarrow  
\sum_{1\leq j\leq r}( \lambda_j-\mu_j) - {\textstyle \frac{r ( |\lambda |-|\mu |)}{n+1}} \in\mathbb{Z}\setminus\mathbb{N}
 \ \text{for}\ r=1,\ldots,n
\end{equation}
(while $\lambda <\mu$ if $\lambda\leq\mu$ and $\lambda\neq\mu$).

\begin{theorem}[Diagonalization]\label{diagonalization:thm}
The following statements hold for parameters from the  regime in  Eq. \eqref{truncation}.

\begin{subequations}
\emph{(i)}
The discrete Ruijsenaars operators $D_1,\ldots ,D_n$ \eqref{Drf} are simultaneously diagonalized in $ \ell^2(\Lambda^{(n,m)},\Delta) $
by an orthogonal basis of joint eigenfunctions.

\emph{(ii)} Upon fixing the normalization such that its value at $\mu=0$ is equal to $1$, an element
$p(\mathbf{e})\in \ell^2(\Lambda^{(n,m)},\Delta) $ of the joint eigenbasis, satisfying
\begin{equation}\label{diagonalization:a}
D_rp(\mathbf{e})=\mathrm{e}_r  p(\mathbf{e})\quad (r=1,\ldots ,n),
\end{equation}
is uniquely determined by the corresponding eigenvalues collected in $\mathbf{e}=(\mathrm{e}_1,\ldots ,\mathrm{e}_n )$.

\emph{(iii)}   The value of the joint eigenfunction  $p (\mathbf{e})$ at $\mu\in\Lambda^{(n,m)}$ is given by
\begin{equation}\label{p}
p_\mu(\mathbf{e})= c_\mu P_\mu (\mathbf{e}) \quad\text{with}\ 
c_\mu = \prod_{1\leq j<k\leq n+1} {\textstyle \frac{[ (k-j)\mathrm{g}]_{\mu_j-\mu_k}}{[ (k-j+1)\mathrm{g}]_{\mu_j-\mu_k}} } .
\end{equation}
Here $P_\mu (\mathbf{e})$ denotes a polynomial in the eigenvalues
$\mathrm{e}_1,\ldots,\mathrm{e}_n$ of the form
\begin{equation}
P_\mu (\mathbf{e}) = \mathrm{e}_\mu+  \sum_{\nu\in\Lambda^{(n,m)},\, \nu < \mu}  u_{\mu ,\nu}\,   \mathrm{e}_\nu \quad\text{with}\
\mathrm{e}_\mu=\prod_{1\leq j\leq n} \mathrm{e}_j^{\mu_j-\mu_{j+1}} ,
\end{equation}
whose expansion coefficients $u_{\mu ,\nu}=u_{\mu ,\nu}(\mathrm{g};p)   \in\mathbb{R}$ are uniquely determined by the recurrence
\begin{equation}
P_\mu (\mathbf{e})= \mathrm{e}_r P_\lambda (\mathbf{e})-
\sum_{\substack{\lambda\subset\nu\subset\lambda+1^{n+1},\,  |\nu|=|\mu| \\ \text{s.t.}\, \underline\nu\in\Lambda^{(n,m)} \setminus \{ \mu\} }}   \psi^\prime_{\nu/\lambda} 
P_{\underline\nu} (\mathbf{e}),
\end{equation}
where $\lambda=\mu-1^r$,
\begin{equation}
r=r_\mu=\min\{ 1\leq j\leq n \mid \mu_j-\mu_{j+1} >0\} ,
\end{equation}
and
\begin{equation}\label{psi}
 \psi^\prime_{\nu/\lambda}  = \prod_{\substack{1\leq j<k\leq n+1\\ \theta_j-\theta_k=-1}} 
 {\textstyle 
  \frac{[\nu_j-\nu_k+\mathrm{g}(k-j+1)]}{[\nu_j-\nu_k+\mathrm{g}(k-j)]}
 \frac{[\lambda_j-\lambda_k+\mathrm{g}(k-j-1)]}{[\lambda_j-\lambda_k+\mathrm{g}(k-j)]}   }
\quad\text{with}\ \theta=\nu-\lambda .
\end{equation}

\emph{(iv)} The polynomials $P_\lambda(\mathbf{e})$, $\lambda\in\Lambda^{(n,m)}$  obey the following Pieri rule on the spectrum
\begin{equation}\label{pieri}
 P_{1^r} (\mathbf{e}) P_\lambda  (\mathbf{e}) = 
\sum_{\substack{\lambda\subset\mu\subset\lambda+1^{n+1},\,  |\mu|=|\lambda|+r \\ \text{s.t.}\, \underline\mu\in\Lambda^{(n,m)}}}   \psi^\prime_{\mu/\lambda} 
P_{\underline\mu}  (\mathbf{e}) \quad\text{for}\  r=1,\ldots ,n.
\end{equation}

\emph{(v)} The joint eigenfunctions  $p(\mathbf{e})$ and $p(\tilde{\mathbf{e}})$ satisfy the orthogonality relation
\begin{equation}
\langle p(\mathbf{e}),p(\tilde{\mathbf{e}} )\rangle_\Delta=0\quad \text{if}\ \mathbf{e}\neq \tilde{\mathbf{e}}.
\end{equation}
\end{subequations}
\end{theorem}

\subsection{Proof of Theorem \ref{diagonalization:thm}}

\emph{(i)} By Proposition \ref{adjoint:prp}, the commuting difference operators $D_1,\ldots,D_n$ \eqref{Drf} are normal
in the $\binom{n+m}{n}$-dimensional space $\ell^2(\Lambda^{(n,m)},\Delta)$. Invoking of the spectral theorem for commuting normal operators in finite dimension
(cf. e.g. \cite[Chapter IX.15]{gan:theory} or \cite[Chapter 2.5]{hor-joh:matrix})
thus suffices to establish the existence of an orthogonal basis of joint eigenfunctions.

\emph{(ii)} \& \emph{(iii)} Let $p\in\ell^2(\Lambda^{(n,m)},\Delta)$ be a joint eigenfunction of $D_1,\ldots, D_n$, which we assume to be normalized such that $p_\lambda=1$ at $\lambda=0$.
In other words, we have that $D_rp=\mathrm{e}_rp$ for some eigenvalue $\mathrm{e}_r\in\mathbb{C}$  ($r=1,\ldots ,n$).
It is immediate from the explicit product formulas in Eqs. \eqref{Dr:c}, \eqref{p} and \eqref{psi} that
for all $\lambda\in\Lambda^{(n,m)}$ and $\lambda\subset \mu \subset \lambda+1^{n+1}$ such that $\underline{\mu}\in\Lambda^{(n,m)}$:
\begin{equation*}
\psi^\prime_{\mu/\lambda}=B_{\mu/\lambda}  \frac{c_{\underline{\mu}} }{ c_\lambda}
\end{equation*}
(where the parameter restriction \eqref{truncation} guarantees that $c_\lambda >0$, cf. the proofs of Lemmas \ref{truncation:lem} and \ref{positivity:lem}). 
The eigenvalue equations for $p$ thus give rise to the following identities for  $P_\lambda=p_\lambda/c_\lambda$ ($\lambda\in\Lambda^{(n,m)}$):
\begin{equation}\label{ev-id}
\mathrm{e}_r P_\lambda= 
\sum_{\substack{\lambda\subset\mu\subset\lambda+1^{n+1},\,  |\mu|=|\lambda|+r \\ \text{s.t.}\, \underline\mu\in\Lambda^{(n,m)}}}   \psi^\prime_{\mu/\lambda} 
P_{\underline\mu} \qquad (r=1,\ldots,n).
\end{equation}

We will now show that  these identities imply that for any $\mu\in\Lambda^{(n,m)}$ the value of $P_\mu$  can be computed uniquely (thus proving \emph{(ii)}) 
in terms of  the polynomial in the eigenvalues $\mathrm{e}_1,\ldots,\mathrm{e}_n$ generated by the recurrence from \emph{(iii)}. 
To this end we perform lexicographical induction in $(d_\mu,r_\mu)$,  where $d_\mu=\mu_1-\mu_{n+1}$ refers to the degree and (recall)
$
r_\mu=\min\{ 1\leq j\leq n \mid \mu_j-\mu_{j+1} >0\}
$ denotes the minimal column size (with the convention that $r_0=0$),
starting from the trivial case that $d_\mu=0$ governed by the initial condition ($d_\mu=0\Rightarrow \mu=0$, so $P_\mu=\mathrm{e}_\mu=1$ in this trivial situation).
Assuming now $d_\mu >0$,  we can write $\mu=\lambda +1^r $ with $r=r_\mu >0$ and  $\lambda\in\Lambda^{(n,m)}$, which implies that $d_\lambda=d_\mu-1$ and $\psi^\prime_{\mu/\lambda}=1$. The induction hypothesis now ensures that
on the LHS of the $r$th relation in  Eq. \eqref{ev-id} the product
$\mathrm{e}_rP_\lambda$ expands as $\mathrm{e}_r\mathrm{e}_\lambda=\mathrm{e}_\mu$ plus a linear combination of monomials of the form
$\mathrm{e}_r\mathrm{e}_\nu=\mathrm{e}_{\nu+1^r}$ with $\nu<\lambda$, i.e. $\nu+1^r<\mu$; the coefficients in this expansion stem from $P_\lambda$
which is generated by the recurrence from \emph{(iii)} (by the induction hypothesis).
The terms on the RHS of the $r$th  relation in Eq. \eqref{ev-id} consist on the other hand of $P_\mu$ plus a linear combination of $P_{\tilde{\mu}}$ with $d_{\tilde{\mu}}\leq d_\mu$ and $\tilde{\mu} <\mu$.
Notice that for the latter terms \emph{either} one has that $d_{\tilde{\mu}}< d_\mu$ \emph{or} one has that $d_{\tilde{\mu}}= d_\mu$ and $\tilde{\mu}=\tilde{\lambda}+1^{\tilde{r}}$ with
$\tilde{\lambda}\in\Lambda^{(n,m)}$, $d_{\tilde{\lambda}}=d_{\mu}-1$, and  $1\leq\tilde{r}<r$. In both cases
it follows from the induction hypothesis that
$P_{\tilde{\mu}}$ is generated by the recurrence from \emph{(iii)} and that
its monomial expansion is given by $\mathrm{e}_{\tilde{\mu}}$ perturbed by a linear combination
of $\mathrm{e}_{\tilde{\nu}}$ with $\tilde{\nu}<\tilde{\mu}$.
Hence, by comparing the expressions on the LHS and the RHS of the $r$th  relation in Eq. \eqref{ev-id}, we see that
$P_\mu$ obeys the recurrence  from  \emph{(iii)} and that we can
express $P_\mu$ as $\mathrm{e}_\mu$  perturbed by a linear combination of monomials $\mathrm{e}_\nu$ with $\nu<\mu$.

\emph{(iv)} The asserted Pieri formula is now immediate from  Eq. \eqref{ev-id} and the observation that $P_{1^r}=\mathrm{e}_{1^r}=\mathrm{e}_r$ for $r=1,\ldots ,n$.

\emph{(v)}
It follows from Proposition \ref{adjoint:prp} that
\begin{equation*}
\mathrm{e}_r=\frac{\langle D_r p(\mathbf{e}),p(\mathbf{e})\rangle_\Delta}{\langle p(\mathbf{e}),p(\mathbf{e})\rangle_\Delta}
=\frac{\langle p(\mathbf{e}), D_{n+1-r} p(\mathbf{e})\rangle_\Delta}{\langle p(\mathbf{e}),p(\mathbf{e})\rangle_\Delta}=\overline{\mathrm{e}}_{n+1-r} ,
\end{equation*}
and thus
\begin{equation*}
\mathrm{e}_r \langle p(\mathbf{e}),p(\tilde{\mathbf{e}} )\rangle_\Delta =
\langle D_r p(\mathbf{e}),p(\tilde{\mathbf{e}})\rangle_\Delta=\langle p(\mathbf{e}), D_{n+1-r} p(\tilde{\mathbf{e}})\rangle_\Delta=\tilde{\mathrm{e}}_r \langle p(\mathbf{e}),p(\tilde{\mathbf{e}} )\rangle_\Delta.
\end{equation*}
Since $\mathrm{e}_r\neq \tilde{\mathrm{e}}_r$ for some $r\in \{ 1,\ldots, n\}$ if $\mathbf{e}\neq\tilde{\mathbf{e}}$, the latter identity requires that
in this situation $ \langle p(\mathbf{e}),p(\tilde{\mathbf{e}} )\rangle_\Delta =0$.

\section{Trigonometric limit}\label{sec5}

\subsection{Macdonald difference operators}
From Eqs. \eqref{scaled-theta}, \eqref{theta} it is immediate that  the scaled theta function $[z;p]$ is analytic in the elliptic nome $p$ for $|p|<1$, while
$[z;0]=\frac{\sin( \alpha z/2 )}{\alpha/2}$.
At $p=0$ the operator $D_r$ \eqref{Drf} therefore simplifies to a finite-dimensional reduction of   Macdonald's difference operator \cite{mac:symmetric,mac:orthogonal}
governed by trigonometric coefficients of the form:
\begin{subequations}
\begin{equation}
B_{\mu/\lambda} = \prod_{1\leq j<k\leq n+1}  {\textstyle \frac{[\lambda_j-\lambda_k+\mathrm{g}(k-j+\theta_j-\theta_k)]_q}{[\lambda_j-\lambda_k+\mathrm{g}(k-j)]_q} }
\quad\text{with}\ \theta=\mu-\lambda ,
\end{equation}
where
\begin{equation}
{\textstyle [z]_q=\frac{\sin( \frac{\alpha z}{2})}{\sin(\frac{\alpha}{2})}= \frac{q^{\frac{z}{2}}-q^{-\frac{z}{2}}}{q^{\frac{1}{2}}-q^{-\frac{1}{2}}}  
 \quad\text{with}\quad  q=e^{\text{i}\alpha} .}
\end{equation}
\end{subequations}
For parameters given by  Eq. \eqref{truncation}, the latter commuting  difference operators are normal
in $\ell^2(\Lambda^{(n,m)},\Delta)$ with
\begin{equation}
\Delta_\lambda =\prod_{1\leq j<k\leq n+1}
{\textstyle
\frac{[\lambda_j-\lambda_k+(k-j)\mathrm{g}]_q}{[(k-j)\mathrm{g} ]_q}\frac{[(k-j+1)\mathrm{g} ]_{q,\lambda_j-\lambda_k}}{[1+(k-j-1)\mathrm{g} ]_{q,\lambda_j-\lambda_k}} } ,
\end{equation}
where
$[z]_{q,k}=\prod_{0\leq l<k} [z+l]_q$ and $ [z]_{q,0}=1$. 
Their spectral decomposition  in
 $\ell^2(\Lambda^{(n,m)},\Delta)$ by means of an orthogonal basis of joint eigenfunctions constructed in terms of Macdonald polynomials
goes back to \cite[Section 4]{die-vin:quantum}.  It is instructive to compare the eigenfunctions  in Theorem \ref{diagonalization:thm} for $p\to 0$ with the ones
from \cite{die-vin:quantum}
given by Macdonald polynomials.

\subsection{Macdonald polynomials}
For $\lambda\in\Lambda^{(n+1)}$ let $P_\lambda(z_1,\ldots ,z_{n+1};q,t)$ denote the Macdonald polynomial \cite[Chapter VI]{mac:symmetric}
with a leading monomial given by
\begin{equation}
m_\lambda (z_1,\ldots,z_{n+1})= \sum_{\nu\in S_{n+1}(\lambda)}  z_1^{\nu_1}\cdots z_{n+1}^{\nu_{n+1}} .
\end{equation}
Here the summation is over all compositions reordering the parts of $\lambda$ (i.e. we sum over the orbit of $\lambda$ with respect to the action of the 
permutation-group $S_{n+1}$ of permutations
$\sigma= { \bigl( \begin{smallmatrix}1& 2& \cdots & n+1 \\
 \sigma_1&\sigma_2&\cdots & \sigma_{n+1}
 \end{smallmatrix}\bigr)}$ on $\lambda_1,\lambda_2,\ldots,\lambda_{n+1}$). The following proposition
 computes the eigenvalues $\mathbf{e}=(\mathrm{e}_1,\ldots,\mathrm{e}_n)$
 in Theorem \ref{diagonalization:thm} explicitly for $p\to 0$ in terms of elementary symmetric polynomials and expresses the corresponding eigenfunctions $p(\mathbf{e})$ in terms of Macdonald polynomials.

\begin{proposition}[Diagonalization at $p=0$]\label{trig-diagonalization:prp}
Let $t=q^{\mathrm{g}}$ and $q=e^{\mathrm{i}\alpha}$ with $\alpha$, $\mathrm{g}$ taken from Eq. \eqref{truncation}. 
The diagonalization and orthogonality from Theorem \ref{diagonalization:thm} can then be rewritten for $p\to 0$ in the following explicit form:
\begin{subequations}
\begin{equation}\label{evM}
D_r p(\mathbf{e}_\nu)= \mathrm{e}_{r,\nu} p(\mathbf{e}_\nu) \quad (r=1,\ldots ,n,\ \nu\in\Lambda^{(n,m)})
\end{equation}
and
\begin{equation}\label{orthoM}
\langle p(\mathbf{e}_\nu),p(\mathbf{e}_{\tilde{\nu}} )\rangle_\Delta=0\quad \text{if}\ \nu \neq \tilde{\nu}\quad (\nu,\tilde{\nu}\in\Lambda^{(n,m)}).
\end{equation}
Here the eigenvalues collected in $\mathbf{e}_\nu=\bigl( \mathrm{e}_{1,\nu},\ldots ,\mathrm{e}_{n,\nu} \bigr)$ 
are  expressed explicitly in terms of the  elementary symmetric polynomials $m_{1^r}$, $r=1,\ldots ,n$:
\begin{equation}
\mathrm{e}_{r,\nu}=
q^{-r \bigl( \frac{|\nu|}{n+1} +\frac{n\mathrm{g}}{2}\bigr)}
m_{1^r} (q^{\nu_1+n\mathrm{g}},q^{\nu_2+(n-1)\mathrm{g}},\ldots, q^{\nu_n+\mathrm{g}},1) ,
\end{equation}
and the value of the corresponding joint eigenfunction
$p(\mathbf{e}_\nu )\in \ell^2(\Lambda^{(n,m)},\Delta)$  at $\mu\in\Lambda^{(n,m)}$ is given by the normalized  Macdonald polynomial:
\begin{equation}
p_\mu (\mathbf{e}_\nu) = c_\mu q^{-|\mu | \bigl( \frac{|\nu|}{n+1}+ \frac{n\mathrm{g}}{2}\bigr)}
P_\mu (q^{\nu_1+n\mathrm{g}},q^{\nu_2+(n-1)\mathrm{g}},\ldots, q^{\nu_n+\mathrm{g}},1 ;q,q^{\mathrm{g}})
\end{equation}
with
\begin{equation}
c_\mu = \prod_{1\leq j<k\leq n+1} {\textstyle \frac{[ (k-j)\mathrm{g}]_{q,\mu_j-\mu_k}}{[ (k-j+1)\mathrm{g}]_{q,\mu_j-\mu_k}} } .
\end{equation}
\end{subequations}
\end{proposition}

\begin{proof}
The orthogonality \eqref{orthoM} follows from a reformulation of the finite-dimensional orthogonality relation for the Macdonald polynomials in  \cite[ Eq. (4.15b)]{die-vin:quantum}
by means of the bijection from Eq. \eqref{bijection} in the appendix below (cf. also \cite[Appendix B]{die-vin:quantum}).
The eigenvalue equation \eqref{evM} amounts in turn to a corresponding reformulation of \cite[Eq. (4.14)]{die-vin:quantum}
(cf. also Eqs. (4.10b), (4.12) and Appendix B of
\cite{die-vin:quantum}).
\end{proof}

\section{Epilogue}\label{sec6}

Because the vector of joint eigenvalues $\mathbf{e}=(\mathrm{e}_1,\ldots ,\mathrm{e}_n)$ in Theorem \ref{diagonalization:thm} is multiplicity free (by part \emph{(ii)}), Proposition \ref{trig-diagonalization:prp} entails
a natural labelling  of the corresponding basis of
joint eigenfunctions for the finite elliptic Ruijsenaars operators $D_1,\ldots,D_n$ \eqref{Drf} in terms of bounded partitions.

\begin{corollary}[Joint Spectrum]\label{spectrum:cor}
The joint spectrum in Theorem \ref{diagonalization:thm} is given by $\binom{n+m}{n}$ vectors
\begin{subequations}
\begin{equation}\label{spectrum:a}
\mathbf{e}_\nu =\bigl( \mathrm{e}_{1,\nu},\ldots ,\mathrm{e}_{n,\nu}\bigr)\in\mathbb{C}^n\quad  (\nu\in\Lambda^{(n,m)})
\end{equation}
that extend analytically to $-1<p<1$, such that
\begin{equation}
\mathrm{e}_{r,\nu} |_{p=0}=
q^{-r \bigl( \frac{|\nu|}{n+1}+\frac{n\mathrm{g}}{2}\bigr)}
m_{1^r} (q^{\nu_1+n\mathrm{g}},q^{\nu_2+(n-1)\mathrm{g}},\ldots, q^{\nu_n+\mathrm{g}},1) 
\end{equation}
$(r=1,\ldots ,n)$.
\end{subequations}
\end{corollary}
Notice at this point that since the operators in question are normal in $\ell^2(\Lambda^{(n,m)},\Delta)$ by Proposition \ref{adjoint:prp},
the analyticity of $\mathbf{e}_\nu$ \eqref{spectrum:a} in $p\in (-1,1)$ is inherited from the analyticity
of  (the coefficients of) $D_r$ \eqref{Drf} and of $\Delta_\lambda$ \eqref{Delta} (cf. \cite[Chapter 2, Theorem 1.10]{kat:perturbation}).

It is now immediate from the orthogonality in Theorem \ref{diagonalization:thm}  that the matrix
$[ \Delta_\mu^{1/2} \hat{\Delta}_\nu^{1/2} p_\mu (\mathbf{e}_\nu )]_{\mu,\nu\in\Lambda^{(n,m)}}$ is unitary, where
\begin{equation}\label{pm}
 \hat{\Delta}_\nu= 1/\langle  p(\mathbf{e}_\nu),p(\mathbf{e}_\nu)  \rangle_\Delta \qquad (\nu\in\Lambda^{(n,m)}).
\end{equation}

\begin{corollary}[Orthogonality]\label{independence:cor}
The polynomials $P_\mu (\mathbf{e})$, $\mu\in\Lambda^{(n,m)}$ form an orthogonal basis for the $\binom{n+m}{n}$-dimensional Hilbert space of (complex)
functions  on the joint spectrum $\{ \mathbf{e}_\nu \mid \nu\in\Lambda^{(n,m)}\} $ associated with the weights $\hat{\Delta}_\nu$ \eqref{pm}:
\begin{equation}
\forall\lambda,\mu\in\Lambda^{(n,m)}:\quad
\sum_{\nu\in\Lambda^{(n,m)}}  P_\lambda (\mathbf{e}_\nu )  \overline { P_\mu (\mathbf{e}_\nu ) } \, \hat{\Delta}_\nu = 
\begin{cases}
\frac{1}{c_\lambda^2 \, \Delta_\lambda} &\text{if}\ \lambda=\mu ,\\
0& \text{if}\ \lambda\neq \mu .
\end{cases}
\end{equation}
\end{corollary}

\appendix

\section{Discretization of Ruijsenaars operators}\label{app}
In this appendix it is outlined how the lattice quantum Ruijsenaars model of Section \ref{sec2} is retrieved from Ruijsenaars' commuting difference operators by discretization.
To this end let us start by recalling that the $\mathfrak{sl}(n+1;\mathbb{C})$ Ruijsenaars operators are commuting difference operators with coefficients built from the Weierstrass $\sigma$-function
(cf. \cite{rui:complete,rui:systems}):
\begin{equation}
\boldsymbol{D}_{r,\sigma }=\sum_{\substack{J\subset\{1,\dots,n+1\}\\|J|=r}}\bigg(\prod_{\substack{j\in J\\k\notin J}}\frac{\sigma(x_j-x_k+ \mathrm{g} )}{\sigma(x_j-x_k)}\bigg) 
T_J,\quad r=1,\ldots ,n.
\end{equation}
Here $\mathrm{g} $ denotes a real coupling parameter and $T_J$ acts by translation on  complex functions $f(\boldsymbol{x})=f(x_1,\ldots,x_{n+1})$:
\begin{equation*}
(T_Jf)(\boldsymbol{x})=f(\boldsymbol{x}+\boldsymbol{\varepsilon}_J)\quad\text{with}\ \boldsymbol{\varepsilon}_J=\sum_{j\in J} \boldsymbol{\varepsilon}_j ,
\end{equation*}
where $\boldsymbol{\varepsilon}_j=\boldsymbol{e}_j-\frac{1}{n+1}(\boldsymbol{e}_1+\cdots +\boldsymbol{e}_{n+1})$ and $\boldsymbol{e}_1,\ldots ,\boldsymbol{e}_{n+1}$ refers to the standard unit basis of $\mathbb{C}^{n+1}$.
A straightforward similarity transformation governed by an appropriate  Gaussian
\begin{equation*}
\boldsymbol{D}_{r}=    c_r G(\boldsymbol{x})^{-1} \boldsymbol{D}_{r,\sigma } G(\boldsymbol{x}) ,
\end{equation*}
with
\begin{equation*}
 c_r = e^{\frac{\zeta (\omega_1)}{2\omega_1}r(n+1-r)\mathrm{g} (1-\mathrm{g} )} \quad \text{and}\quad
   G(\boldsymbol{x})=\exp\Bigl( {\textstyle -\frac{\mathrm{g} \, \zeta(\omega_1) }{2\omega_1}} \sum_{1\leq j<k\leq n+1} (x_j-x_k)^2\Bigr),
\end{equation*}
recasts these difference operators into the form (cf. Eq. \eqref{jw-conversion}):
\begin{equation}\label{Dr}
\boldsymbol{D}_{r}=\sum_{\substack{J\subset\{1,\dots,n+1\}\\|J|=r}} V_J(\boldsymbol{x})
T_J \quad \text{with}\ V_J(\boldsymbol{x})= \prod_{\substack{j\in J\\k\notin J}}\frac{[x_j-x_k+ \mathrm{g} ]}{[x_j-x_k]}  .
\end{equation}

We will now discretize the Ruijsenaars operator \eqref{Dr} on a translate of the $\mathfrak{sl}(n+1;\mathbb{C})$ dominant weight lattice
\begin{equation}\label{dw}
\boldsymbol{\Lambda}^{(n)}= \{ l_1\boldsymbol{\omega}_1+\cdots +l_n\boldsymbol{\omega}_n \mid  l_1,\ldots,l_n\in\mathbb{Z}_{\geq 0}  \} ,
\end{equation}
which is generated by the corresponding fundamental weights $\boldsymbol{\omega}_r=\boldsymbol{\varepsilon}_1+\cdots+\boldsymbol{\varepsilon}_r$ ($r=1,\ldots,n$). The pertinent translation is over a $\mathrm{g} $-deformation of the Weyl vector
\begin{equation}
\boldsymbol{ \rho}_\mathrm{g} =\mathrm{g} (\boldsymbol{\omega}_1+\cdots+\boldsymbol{\omega}_n) =\mathrm{g}   \sum_{1\leq j\leq n+1}  \bigl({\textstyle \frac{n}{2}}+1-j\bigr)\boldsymbol{e}_j .
\end{equation}
To avoid singularities stemming from the denominator of $V_J(\boldsymbol{x})$ \eqref{Dr}, it will from now on be assumed that $\mathrm{g} \in\mathbb{R}$ is generic such that
\begin{equation}\label{generic}
\forall\boldsymbol{\lambda}\in\boldsymbol{\Lambda}^{(n)}:\quad  \prod_{1\leq j <k \leq n+1}  \left.
{\textstyle  \sin\frac{\alpha}{2}(x_j-x_k) } \right|_{\boldsymbol{x}=\boldsymbol{\rho}_\mathrm{g} +\boldsymbol{\lambda}} \neq 0.
\end{equation}
The next lemma confirms that $\boldsymbol{D}_r$ \eqref{Dr} restricts in this situation to a discrete difference operator acting on lattice functions
$f:(\boldsymbol{\rho}_\mathrm{g} +\boldsymbol{\Lambda}^{(n)})\to\mathbb{C}$.

\begin{lemma}\label{discretization:lem}
For $\boldsymbol{\lambda} \in\boldsymbol{\Lambda}^{(n)}$ and $J\subset \{ 1,\ldots ,n+1\}$ one has that
\begin{equation}
V_J(\boldsymbol{\rho}_\mathrm{g} +\boldsymbol{\lambda})=0\quad\text{if}\ \boldsymbol{\lambda}+\boldsymbol{\varepsilon}_J\not\in \boldsymbol{\Lambda}^{(n)} .
\end{equation}
\end{lemma}
\begin{proof}
Let us recall that the dominant cone $\boldsymbol{\Lambda}^{(n)}$ \eqref{dw}  constitutes a fundamental domain for the $\mathfrak{sl}(n+1;\mathbb{C})$ weight lattice $ \{ l_1\boldsymbol{\omega}_1+\cdots +l_n\boldsymbol{\omega}_{n} \mid  l_1,\ldots,l_{n}\in\mathbb{Z}  \} $
with respect to the action of the permutation group $S_{n+1}$ on the unit basis $\boldsymbol{e}_1,\ldots,\boldsymbol{e}_{n+1}$. More specifically, a vector $\boldsymbol{\lambda}=\sum_{1\leq j\leq n+1}\lambda_j \boldsymbol{e}_j$ in the weight lattice belongs to
$\boldsymbol{\Lambda}^{(n)}$ iff
$\lambda_1\geq\lambda_2\geq\cdots \geq \lambda_{n+1}$. Hence,  if $\boldsymbol{\lambda} \in\boldsymbol{\Lambda}^{(n)}$ and $\boldsymbol{\mu}=\boldsymbol{\lambda}+\boldsymbol{\varepsilon}_J\not\in \boldsymbol{\Lambda}^{(n)}$
then $\mu_j-\mu_{j+1}<0$ for some $j\in \{1,\ldots,n\}$, which implies that  $j\not\in J$,  $j+1\in J$ and $\lambda_j-\lambda_{j+1}=0$. We thus pick up a zero
of $V_J(\boldsymbol{x})$ at $\boldsymbol{x}=\boldsymbol{\rho}_\mathrm{g} +\boldsymbol{\lambda}$ from the factor $[x_{j+1}-x_{j}+\mathrm{g} ]$ in this situation.
\end{proof}

With the aid of Lemma \ref{discretization:lem} we see that the discretized Ruijsenaars
operator
\begin{equation}\label{Dr:discrete}
(\boldsymbol{D}_{r}f)(\boldsymbol{\rho}_\mathrm{g} +\boldsymbol{\lambda})=\sum_{\substack{J\subset\{1,\dots,n+1\}\\|J|=r,\, \boldsymbol{\lambda}+\boldsymbol{\varepsilon}_J\in\boldsymbol{\Lambda}^{(n)}}} V_J(\boldsymbol{\rho}_\mathrm{g} +\boldsymbol{\lambda}) f(\boldsymbol{\rho}_\mathrm{g} +\boldsymbol{\lambda}+\boldsymbol{\varepsilon}_J) 
\end{equation}
gives rise to commuting difference operators $\boldsymbol{D}_1,\ldots ,\boldsymbol{D}_n$ in the space of lattice functions $f:(\boldsymbol{\rho}_\mathrm{g} +\boldsymbol{\Lambda}^{(n)})\to\mathbb{C}$.

The discrete Ruijsenaars operator $D_r$ \eqref{Dr:a}--\eqref{Dr:c}  boils down to a reformulation of $\boldsymbol{D}_r$ \eqref{Dr:discrete} in terms of partitions, via the bijection 
\begin{equation}\label{bijection}
\boldsymbol{\lambda}=l_1\boldsymbol{\omega}_1+\cdots +l_n\boldsymbol{\omega}_n \leftrightarrow (l_1+\cdots+l_n,l_2+\cdots+l_n,\ldots,l_{n-1}+l_n,l_n) =\lambda
\end{equation}
identifying the dominant weight lattice $\boldsymbol{\Lambda}^{(n)}$ \eqref{dw} with the lattice $\Lambda^{(n)}$  \eqref{partitions:n} of partitions of length at most $n$.
This bijection maps the
action of the discrete Ruijsenaars operator $\boldsymbol{D}_r$ \eqref{Dr:discrete} on $f:(\boldsymbol{\rho}_\mathrm{g} +\boldsymbol{\Lambda}^{(n)})\to\mathbb{C}$ to that
of $D_r$ \eqref{Dr:a}--\eqref{Dr:c} on the lattice function $\lambda\stackrel{f}{\to} f_\lambda$, $\lambda\in\Lambda^{(n)}$
via the dictionary
\begin{equation}
f(\boldsymbol{\rho}_\mathrm{g} +\boldsymbol{\lambda})=f_\lambda\quad\text{and}\quad \theta_j=\begin{cases} 1& \text{if}\  j\in J\\ 0&\text{if}\ j\not\in J 
\end{cases} .
\end{equation}
Indeed, with these identifications one has  that
\begin{equation}
\boldsymbol{\lambda}+\boldsymbol{\varepsilon}_J\leftrightarrow \underline{\mu}  \quad\text{and}\quad V_J(\boldsymbol{\rho}_\mathrm{g} +\boldsymbol{\lambda}) = B_{\mu/\lambda}\quad \text{where}\ \mu-\lambda=\theta .
\end{equation}
The lattice operators  $D_r$ \eqref{Dr:a}--\eqref{Dr:c}  thus inherit the commutativity from $\boldsymbol{D}_r$ \eqref{Dr:discrete}.

\section*{Acknowledgements}
The work of JFvD was supported in part by the {\em Fondo Nacional de Desarrollo
Cient\'{\i}fico y Tecnol\'ogico (FONDECYT)} Grant \# 1210015. TG was supported in part by the NKFIH Grant K134946.

\bigskip\noindent
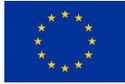
\parbox{.135\textwidth}{\begin{tikzpicture}[scale=.03]
\fill[fill={rgb,255:red,0;green,51;blue,153}] (-27,-18) rectangle (27,18);  
\pgfmathsetmacro\inr{tan(36)/cos(18)}
\foreach \i in {0,1,...,11} {
\begin{scope}[shift={(30*\i:12)}]
\fill[fill={rgb,255:red,255;green,204;blue,0}] (90:2)
\foreach \x in {0,1,...,4} { -- (90+72*\x:2) -- (126+72*\x:\inr) };
\end{scope}}
\end{tikzpicture}} \parbox{.85\textwidth}{This project has received funding from the European Union's Horizon 2020 research and innovation programme under the Marie Sk{\l}odowska-Curie grant agreement No 795471.}

\bibliographystyle{amsplain}

\end{document}